\pdfoutput=1
\documentclass{article}
\usepackage[usenames, dvipsnames]{xcolor}
\usepackage[utf8]{inputenc}
\usepackage[english]{babel}
\usepackage{amsmath,amsfonts}
\usepackage{amssymb}
\usepackage[textwidth=2in]{todonotes}
\usepackage{mathtools}
\usepackage{amsthm}
\usepackage{moresize}
\usepackage{tikz}
\usepackage[letterpaper, portrait, left=1.3in, right=1.3in, top=1.0in, bottom=1.2in]{geometry}
\usepackage[numbers,sectionbib]{natbib}
\usepackage[hyperindex,breaklinks]{hyperref}
\usepackage[ruled,vlined, linesnumbered]{algorithm2e}
\usepackage{cleveref}
\usepackage[shortlabels]{enumitem}
\usepackage{booktabs} 
\usepackage{xparse} 
\usepackage{array}
\usepackage{booktabs}
\usepackage{pifont}
\usepackage[toc]{appendix}
\usepackage{autonum}
\usepackage{makecell}
\usepackage{float}
\usepackage{tikz}

\makeatletter
\newcommand\incircbin
{%
  \mathpalette\@incircbin
}
\newcommand\@incircbin[2]
{%
  \mathbin%
  {%
    \ooalign{\hidewidth$#1#2$\hidewidth\crcr$#1\bigcirc$}%
  }%
}

\makeatother

\SetKwRepeat{Repeat}{repeat}{}

\usetikzlibrary{patterns}

\def\squareforqed{\leavevmode\hbox to.77778em{\hfil\vrule\vbox to.675em{\hrule width.6em\vfil\hrule}\vrule\hfil}} 

\newtheorem{theorem}{Theorem}
\newtheorem{corollary}[theorem]{Corollary}

\newtheorem{lemma}[theorem]{Lemma}

\newtheorem{proposition}[theorem]{Proposition}

\theoremstyle{remark}

\newcommand*{\vcenteredhbox}[1]{\begingroup
\setbox0=\hbox{#1}\parbox{\wd0}{\box0}\endgroup}

\SetKwFor{RepTimes}{repeat}{times}{end}

\newcommand{\nosemic}{\renewcommand{\@endalgocfline}{\relax}}
\newcommand{\dosemic}{\renewcommand{\@endalgocfline}{\algocf@endline}}

\newcommand{\tmax}{\tau}
\newcommand{\wask}{}

\SetKwFor{RepTimes}{repeat}{times}{end}

\date{}
\title{Approximate Triangle Counting via Sampling and Fast Matrix Multiplication}
\author{Jakub Tětek\\\texttt{j.tetek@gmail.com}\\Basic Algorithms Research Copenhagen,\\ University of Copenhagen}
\begin{document}
\maketitle
\begin{abstract}
There is a trivial $O(\frac{n^3}{T})$ time algorithm for approximate triangle counting where $T$ is the number of triangles in the graph and $n$ the number of vertices. At the same time, one may count triangles exactly using fast matrix multiplication in time $\tilde{O}(n^\omega)$. Is it possible to get a negative dependency on the number of triangles $T$ while retaining the $n^\omega$ dependency on $n$? We answer this question positively by providing an algorithm which runs in time $O\big(\frac{n^\omega}{T^{\omega - 2}}\big) \cdot \text{poly}(n^{o(1)}/\epsilon)$. This is optimal in the sense that as long as the exponent of $T$ is independent of $n, T$, it cannot be improved while retaining the dependency on $n$; this as follows from the lower bound of Eden and Rosenbaum [APPROX/RANDOM 2018]. Our algorithm improves upon the state of the art when $T = \omega(1)$ and $T = o(n)$.

We also consider the problem of approximate triangle counting in sparse graphs, parameterizing by the number of edges $m$. The best known algorithm runs in time $\tilde{O}\big(\frac{m^{3/2}}{T}\big)$ [Eden et al., SIAM Journal on Computing, 2017]. There is also a well known algorithm for exact triangle counting that runs in time $\tilde{O}(m^{2\omega/(\omega + 1)})$. We again get an algorithm that retains the exponent of $m$ while running faster on graphs with larger number of triangles. Specifically, our algorithm runs in time $O\Big(\frac{m^{2\omega/(\omega+1)}}{ T^{2(\omega-1)/(\omega+1)}}\Big) \cdot \text{poly}(n^{o(1)}/\epsilon)$. This is again optimal in the sense that if the exponent of $T$ is to be constant, it cannot be improved without worsening the dependency on $m$. This  algorithm improves upon the state of the art when $T = \omega(1)$ and $T = o(\sqrt{m})$.
\end{abstract}

\section{Introduction}
The problem of counting triangles in a graph is both a fundamental problem in graph algorithms and a problem with many applications, for example in network science \cite{Barabasi2016}, biology \cite{Chung2006} or sociology \cite{Wimmer2010}. It is, also for these reasons, one of the basic procedures in graph mining. Consequently, triangle counting has received a lot of attention both in the theoretical and applied communities.

There are algorithms for \emph{exact} triangle counting that run in time $\tilde{O}(n^\omega)$ \footnote{We use $\tilde{O}$ with the meaning that it ignores an $n^{o(1)}$ factor; this is necessary as $\omega$ only determines the time complexity up to $n^{o(1)}$. $\tilde{O}_\epsilon$ in addition ignores $\text{poly}(\frac{1}{\epsilon})$. This differs from the usual use which only ignores $\text{poly} \log n$ factors.} and $\tilde{O}(m^{2\omega/(\omega+1)})$ \cite{Alon1997} where $\omega$ is the smallest constant such that two $n\times n$ matrices can be multiplied in time $n^{\omega + o(1)}$. There are also algorithms for \emph{approximate} triangle counting that have \emph{negative dependence} on the number of triangles $T$. In the setting where the algorithm may sample random vertices and edges, there are algorithms that run in time 
$O(\frac{n^3}{T})$ and $\tilde{O}(\frac{m^{3/2}}{T})$.
However, no algorithm is known that would get the best of both worlds --- algorithm with state-of-the-art dependency on $n$ or $m$ while being negatively dependent on $T$. We show two such algorithms in this paper. Moreover, our algorithms are optimal in the sense that the exponent of $T$ in the time complexity of our algorithms cannot be improved without worsening the dependency on $n$ or $m$ as long as the exponent of $T$ is a constant (i.e., is independent of $n,m,T$).
This optimality follows from the lower bound shown by \citet{Eden2017}, as we will discuss shortly.

The main contribution of our paper is an algorithm for approximate triangle counting that runs in time $\tilde{O}_\epsilon\big(\frac{n^\omega}{T^{\omega-2}}\big)$. We then use it to give an algorithm for the same problem running in time $\tilde{O}_\epsilon\Big(\frac{m^{2\omega/(\omega+1)}}{T^{2(\omega-1)/(\omega+1)}}\Big)$. This improves upon the first algorithm for sparse graphs. Note that both time complexities are sublinear for sufficiently large value of $T$. To the best of our knowledge, this is the first work which uses fast matrix multiplication in sublinear-time algorithms. Assuming constant $\epsilon$, our algorithms improve upon the state-of-the-art when $T = \Omega(n^\delta)$ for some $\delta > 0$, and $T \leq o(n)$ or $T = \Omega(m^\delta)$ and $T \leq o(\sqrt{m})$ (for \Cref{alg:tc_dense} and \Cref{alg:tc_sparse}, respectively). In other words, our algorithms improve upon the state of the art when the time complexity (of both our and the state-of-the-art algorithms) is $\omega(n^2)$ and $\omega(m)$, respectively, for constant $\epsilon$. For diagrammatic comparison of our algorithms with previous work, see \Cref{fig:time_complexities}.

The basic approach we use is to sample vertices, recursively count triangles in the subgraph induced by these sampled vertices, and estimate the total number of triangles based on that. The main hurdle in this approach is that when a vertex contains many triangles, this results in a poor concentration of the number of triangles in the induced subgraph as the inclusion of one vertex can make a large difference in the number of triangles. We get around this by introducing a procedure based on fast matrix multiplication which finds all vertices contained in many triangles, and a procedure for approximate counting triangles containing such vertices. This allows us to reduce the problem of counting triangles in a graph to the problem of counting triangles in a smaller graph; we solve this problem recursively. An obstacle to overcome is that in some cases it may happen that the time spent in each successive level of recursion grows exponentially.

\begin{figure}
\includegraphics[width=\textwidth]{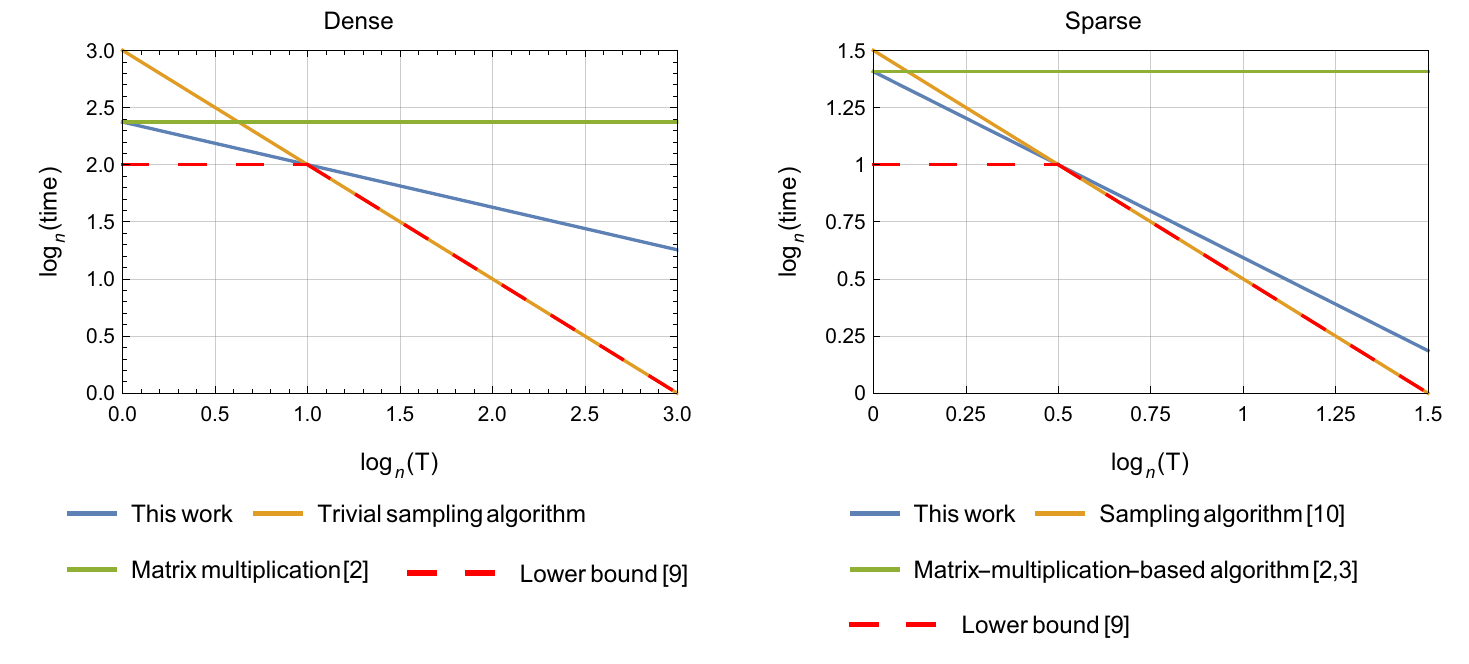}
\caption{Comparison of running times, assuming $\epsilon$ is constant. This plot assumes the best currently known bound $\omega < 2.3728596$ from \citet{Alman2021}. For more details, see \Cref{tab:summary_results}.} \label{fig:time_complexities}
\end{figure}

\paragraph{The setting.} When an algorithm runs in sublinear time, it is of importance how the algorithm is allowed to access the graph. This is the case because the algorithm does not have the time to pre-process the whole graph. We assume the following standard setting: the algorithm may in constant time (1) get the $i$-th vertex, (2) get the $j$-th neighbor of a given vertex, (3) query the degree of a given vertex; in the second part of this paper, we assume the setting which in addition allows the algorithm to get the $i$-th edge of the graph in constant time. This setting is also standard in the area of sublinear algorithms. As the main contribution of this paper is the case with superlinear complexity, we do not describe the settings in more details. See for example \cite{Assadi2018} for a more detailed description.
%
%

\paragraph{Time complexity vs. query complexity}
In the area of sublinear algorithms, it is customary to focus on the \emph{query complexity} (that is, the number of queries in the above sense performed) of an algorithm, as opposed to the time complexity. The time complexity is then usually near-linear in the query complexity. However, this is the case usually only in the sublinear regime, and it breaks down when the time complexity is superlinear. In this case, one may always read the whole graph in $O(n+m)$ queries. The time complexity is then significantly larger than the query complexity. This is the regime that is our main focus. 

Specifically, the algorithm from \cite{Eden2017} are known to be near-optimal\footnote{Whenever we talk about optimality, we consider the case of $\epsilon$ being a constant.} in terms of the query complexity. In the sublinear regime (that is, for $T$ large enough), the asymptotic time complexity is the same as the query complexity and it is, therefore, also optimal. However, in the superlinear regime, the time complexity is polynomially larger than the query complexity. Up until now, nothing was known about whether the time complexity can be improved in this case. We answer this question positively. In doing this, we give, to the best of our knowledge, the first algorithm answering a question of this type (that is, a question regarding the discrepancy between the time complexity and query complexity in the superlinear regime).

\subsection{Our results}
The main contributions of this paper are the following two algorithms for approximate triangle counting. The first algorithm (\Cref{alg:tc_dense}) is suitable for dense graphs and its time complexity is is parameterized by $n$ and $T$; the second algorithm (\Cref{alg:tc_sparse}) is suitable for sparse graphs and its time complexity is parameterized by $m$ and $T$. Both of our algorithms have a state-of-the-art dependency on $n,m$ while being negatively dependent on $T$. Our algorithms improve upon the existing algorithms when $T = o(n)$ and $T = o(\sqrt{m})$, respectively.

We compare these results to the previous work in the following table. These time complexities are also plotted in \Cref{fig:time_complexities}. See the description of the plot for details. Note that the plot ignores dependency on $n$ in the case of sparse triangle counting. We view this as inconsequential as our algorithm improves upon the state-of-the-art for the case when there are relatively few triangles. On this range, $n$ is indeed a lower-order term, unless there are more vertices than edges. 

\def\arraystretch{1.2}
\setlength{\tabcolsep}{1.5em}
\begin{table}[H]
\centering
\begin{tabular}{|l|l|l|}
\hline& Dense & Sparse \\\hline
This work & $\tilde{O}_\epsilon\big(\frac{n^\omega}{T^{\omega-2}}\big)$ & $\tilde{O}_\epsilon\Big(\frac{m^{2\omega/(\omega+1)}}{T^{2(\omega-1)/(\omega+1)}}\Big) $  \\
Trivial algorithm & $O_\epsilon(\frac{n^3}{T})$ &   \\
Fast matrix multiplication & $\tilde{O}(n^\omega)$ &   \\
\citet{Alon1997} &  & $\tilde{O}(m^{2\omega/(\omega+1)})$ \\
\citet{Eden2017} & $\Omega\big(\min(n^2, \frac{n^3}{T})\big)$ & $\tilde{O}_\epsilon( \frac{m^{3/2}}{T})$, $\Omega\big(\min(m, \frac{m^{3/2}}{T})\big)$ \\\hline
\end{tabular}
\caption{Algorithms for (approximate) triangle counting.}
\label{tab:summary_results}
\end{table}
\subsection{Related work}
\paragraph{Algorithms.} The number of triangles in a graph can be trivially counted in time $\tilde{O}(n^3)$. \citet{Itai1977} obtained an algorithm that runs in time $O(m^{3/2})$, a significant improvement for sparse graphs. A simpler algorithm based on bounds on arboricity has been given by \citet{Chiba1985} that implies the same bounds. A more efficient algorithm for \emph{approximate} triangle counting has been obtained by \citet{Kolountzakis2012}. In their paper, the authors presented an algorithm that runs in time $O(m+\frac{m^{3/2} \log n}{T \epsilon^2})$. This has been later improved by \citet{Eden2017} to $\tilde{O}(\frac{n}{T^{1/3}} + \frac{m^{3/2}}{T})$ when the random edge queries are not allowed and to $\tilde{O}(\frac{m^{3/2}}{T})$ when they are.

While the fast matrix multiplication gives a $\tilde{O}(n^\omega)$ algorithm for exact triangle counting, it is not clear an improvement can be achieved in sparse graphs. In their paper, \citet{Alon1997} show an algorithm for exact triangle counting that runs in time $\tilde{O}(m^{2\omega/(\omega+1)})$. This algorithm works by counting triangles whose all vertices have high degree using matrix multiplication while using a naïve algorithm for the rest of the graph. We use a variant of this approach for approximate triangle counting in a sparse graph in \Cref{sec:sparse}.

Sampling from graphs in order to estimate the number of triangles has been considered many times \cite{Tsourakakis2009,Pagh2012,Kallaugher2016,Biswas2020}; for more details, see the introduction of \cite{Kallaugher2016}. The bounds then often depend on an upper bound on the number of triangles containing any single edge or similar parameters, for example in \cite{Pagh2012} or \cite{Tsourakakis2009}. In this paper, we go further in the sense that we use a similar sampling approach (in fact, we use the sampling scheme from \cite{Biswas2020} with finer analysis) but explicitly ensure the bound on the number of triangles containing any vertex. In other papers using similar sampling schemes, it is just assumed that this or a similar assumption is satisfied for the input instance.

\paragraph{Lower bounds.} From the side of lower bounds, \citet{Eden2017} have proven a lower bound on the query complexity, thus also implying a lower bound on the time complexity. This lower bound is presented for the setting when random edge queries are not allowed but the lower bound of $\Omega(\min(m, m^{3/2}/T))$ also holds in the setting when they are. This was made explicit by \citet{Eden2018b} who presented a significantly simpler proof based on communication complexity. We use this lower bound to prove the claim that the exponent of $T$ in \Cref{thm:tc_dense} is optimal, out of all algorithms running in time $\tilde{O}(n^\omega/T^c)$ for some constant $c$ and similarly the exponent in \Cref{thm:tc_sparse} is optimal out of all algorithms running in time $\tilde{O}(m^{2\omega/(\omega+1)}/T^c)$. Specifically, the following theorem is proven in \cite{Eden2018b}:
\begin{theorem}[Theorem 4.7 in \citet{Eden2018b}]
For any $n, m \in O(n^2), T \in O(m^{3/2})$, it holds that any multiplicative-approximation algorithm for the number of triangles in a graph must perform $\Omega\left(\min \left\{m, \frac{m^{3 / 2}}{T}\right\}\right)$ queries, where the allowed queries are degree queries, pair queries\footnote{By pair query, we mean the query ``are vertices $u,v$ adjacent?"}, random neighbor queries, random vertex queries and random edge queries.
\end{theorem}
By choosing $m = \Theta(n^2)$, it follows that
\begin{corollary}
For any $n, T \in O(n^3)$, it holds that any multiplicative-approximation algorithm for the number of triangles in a graph must perform $\Omega\left(\min \left\{n^2, \frac{n^3}{T}\right\}\right)$ queries, where the allowed queries are degree queries, pair queries, neighbor queries, random vertex queries and random edge queries.
\end{corollary}

\paragraph{Fast matrix multiplication.} There is a large literature on fast matrix multiplication. We only mention here the most efficient matrix multiplication algorithm which is currently the algorithm by \citet{Alman2021} which gives $\omega < 2.3728596$.

\subsection{Technical overview}
\paragraph{Triangle counting in dense graphs.} Suppose we sample each vertex independently with some probability $p$. The expected number of triangles remaining is then $p^3 T$. If we were able to show concentration around the expected value, we could sample a sufficiently large subgraph, count the number of triangles in that subgraph and based on that, estimate the number of triangles in the original graph.

Unfortunately, the number of triangles is not concentrated; for example in the case when there is one vertex which is contained in all triangles. We show that we get concentration if we limit the number of triangles containing any single vertex. We, therefore, have some threshold $\tau$ and call all vertices that are contained in more than $\tau$ triangles triangle-heavy. Assume that the graph does not have any triangle-heavy vertices. By choosing $\tau$ sufficiently small and by sampling sufficiently many vertices, we may get a good estimate of the number of triangles in the original graph, based on the number of triangles in the subgraph induced by the sampled vertices.

To be able to use this, we have to be able to find all triangle-heavy vertices in the graph and then estimate the number of triangles that contain at least one triangle-heavy vertex. It is this procedure for finding the triangle-heavy vertices that uses matrix multiplication.

The last trick we introduce is that we use recursion to approximately count the triangles in the sampled subgraph. An obstacle is that the precision increases in the depth of recursion. Specifically, in our algorithm, the allowed error decreases at an exponential rate. Moreover, we have to use probability amplification, meaning that the number of calls grows exponentially with the depth of recursion. This leads, in some situations, to the time complexity of the $k$-th level of recursion increasing exponentially in $k$; we bound this time.

\paragraph{Triangle counting in sparse graphs} We set a parameter $\theta$ and say a vertex is degree-heavy if its degree is at least $\theta$ and degree-light otherwise. We count separately \emph{(a)} triangles on the subgraph induced by the degree-heavy vertices and \emph{(b)} triangles that contain at least one degree-light vertex. This is basically the idea used by \citet{Alon1997} for their exact triangle-counting algorithm. We use our algorithm for approximate triangle counting in dense graphs to count the triangles from case \emph{(a)} and sampling to count the triangles \emph{(b)}. The sampling-based estimation is based on ideas from \cite{Chiba1985,Eden2017}. Specifically, we assign each triangle to its vertex with the lowest degree (breaking ties arbitrarily). We then repeatedly perform the following experiment: pick an edge uniformly at random; consider its endpoint $v$ with the lower degree, breaking ties arbitrarily; pick one of its neighbors at random; check whether it forms a triangle together with the edge we have sampled and whether the triangle is assigned to $v$. We estimate the number of triangles based on the proportion of experiments that ended with finding a triangle.

\subsection{Notation}
Throughout the paper, we denote the number of vertices and edges of a graph by $n$ and $m$, respectively. We use $T$ to denote the number of triangles. We define $a \pm b = [a - b, a+b]$. This can be used for example in $(1\pm \epsilon) a$ which is then equal to $[(1-\epsilon)a, (1+\epsilon)a]$. We call $(1\pm \epsilon)$-approximation to a number $a$ any number $b$ such that $b \in (1\pm \epsilon)a$. Similarly, we call \emph{additive} $\pm c$ approximation to a number $a$ any number $b$ such that $b \in a \pm c$. We call \emph{diamond} a graph isomorphic to \vcenteredhbox{\includegraphics[height=2\fontcharht\font`\B]{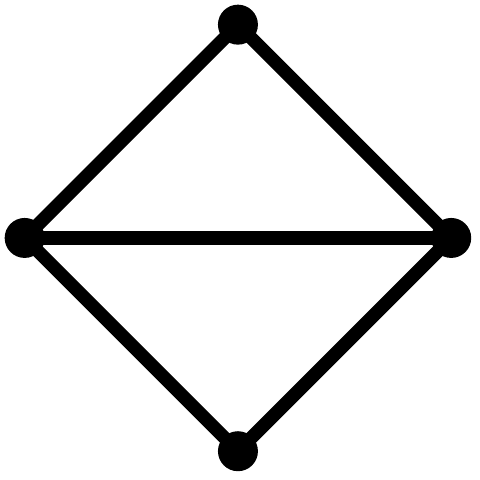}} and \emph{butterfly} a graph isomorphic to \vcenteredhbox{\includegraphics[height=2\fontcharht\font`\B]{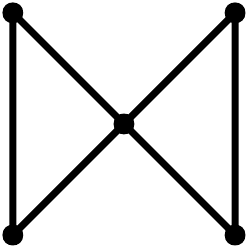}}. By ``the number of diamonds in $G$'' we mean the number of, not necessarily induced, subgraphs of $G$ ismomorphic to a diamond. We similarly talk about ``the number of butterflies in a graph''. We denote the subgraph of $G$ induced by $V' \subseteq V$ by $G[V']$.

\section{Triangle counting in dense graphs} \label{sec:dense}
In this section, we present the main result of our paper --- an algorithm for approximate triangle counting, parameterized by $n$ and the number of triangles $T$. In the lemmas that follow, we prove among other things, bounds on the expectation of the estimators. We do this because our algorithm requires some ``advice" and having a bound on the expectation will allow us to remove the need for this advice. We now define some notation that we will need.

We call a vertex \emph{triangle-dense} if there are at least $\tmax$ triangles that contain the vertex for some parameter $\tmax$ and \emph{triangle-light} if there at most $\tmax/20$ such triangles (note that there can be vertices that are neither triangle-light nor triangle-heavy). We define $T_v$ to be the number of triangles containing the vertex $v$. We abuse notation and also denote by $T_v$ the set of all triangles containing $v$. We use $T_v(G)$ when we want to explicitly specify the graph. $T$ denotes the number of triangles in $G$ but we also abuse notation and use this to denote the set of triangles. Again, if we want to specify the graph, we use $T(G)$.

We now prove a simple lemma on fractional moments of the binomial distribution. We will use this to bound the time complexity of matrix multiplication executed on subgraph obtained by keeping each vertex with some appropriately chosen probability and removing the other vertices.
\begin{lemma} \label{lem:binom_moment}
Assume $p \geq 1/n$. It holds that $E(Bin(n,p)^\omega) = O(n^\omega p^\omega)$
\end{lemma}
\begin{proof}
$Bin(n,p)$ has its third moment equal to
\[
(n(n-1)(n-2)p^3) + 3(n(n-1)p^2) + np = O(n^3p^3)
\]
where the equality holds because $p\geq1/n$. Since $\omega < 3$, the function $x^{\omega/3}$ is concave. By the Jensen's inequality, we then get
\[
E(Bin(n,p)^\omega) \leq (E(Bin(n,p)^3))^{\omega/3} \leq O(n^\omega p^\omega)
\]
\end{proof}

\subsection{Triangle-light subgraph}
We show a reduction from the problem of approximately counting triangles in a graph with no triangle-heavy vertices to the problem of approximately counting triangles in a smaller graph. Before that, we need to prove the following lemma.
\begin{lemma} \label{lem:diamond_butterfly_bound}
Let $G$ be a graph with $T$ triangles such that any vertex is contained in at most $\tau$ triangles and let $T,B$ be the diamonds and butterflies in $G$, respectively. Then $D+B \leq \frac{3}{2} \tau T$.
\end{lemma}
\begin{proof}
Let $D,B$ be the number of diamonds and butterflies in $G$ and let $S$ be the number of pairs $(\triangle_1,\triangle_2)$ of triangles in $G$ such that $\triangle_1$ and $\triangle_2$ share at least one vertex. It holds $D+B \leq S/2$ (the factor of $2$ is there because the pairs are ordered). Consider a fixed triangle $\triangle_1 = uvw$. How many triangles are there that share a vertex with $\triangle$? Each of the vertices $u,v,w$ can be contained in at most $\tau - 1$ other triangles. There can be, therefore, at most $3(\tau-1)$ incidences between $\triangle$ and other triangles containing any of the vertices $u,v,w$. Therefore $D+B \leq S/2 \leq 3T(\tau-1)/2$
\end{proof}
%
%


\begin{lemma} \label{lem:sampling_concentration}
Let $G' = (V',E')$ be the induced subgraph of $G$ resulting from selecting each vertex with probability $p = \max(5 \frac{1}{\sqrt[3]{\epsilon^2 T}}, 10 \frac{1}{\epsilon \sqrt{T/\tmax}})$ and removing it otherwise. Let $\hat{T} = |T(G')|/p^3$. Then $\hat{T}$ is an unbiased estimate of $T$. Assume any vertex $v \in G$ is contained in at most $\tmax$ triangles. Then $\hat{T} \in (1\pm\epsilon) T$ with probability at least $19/20$.
\end{lemma}
\begin{proof}
Let $\triangle$ be a triangle in $G$. Let $X_{\triangle} = 1$ if $\triangle \in T(G')$ and $0$ otherwise. Let $T' = T(G')$. It holds that
\[
E(\hat{T}) = E(T'/p^3) = \sum_{\triangle \in T}E(X_\triangle)/p^3 = T
\]
and the estimate is, therefore, unbiased.

We now bound $Var(T')$. Recall that $T(G)$ is the set of triangles in $G$. Let $D(G)$ the set of diamonds in $G$, $B(G)$ the set of butterflies in $G$, and $\ddot{T}(G)$ the set of disjoint (unordered) pairs of triangles in $G$. Let $D = |D(G)|$, $B = |B(G)|$, and $\ddot{T} = |\ddot{T}(G)|$. By $\{\triangle_1, \triangle_2\} \in D(G)$ we mean that $\triangle_1 \cup \triangle_2 = A$ for some $A \in D(G)$ and analogously for $B(G)$ and $\ddot{T}(G)$.


We first bound the second moment of $T'$
\begin{align}
E({T'}^2)=& E\Bigg(\Big(\sum_{\triangle \in T(G)} X_\triangle \Big)^2\Bigg)\\
=& E\Big(\sum_{\triangle \in T(G)} X_\triangle^2\Big) + E\Big(\sum_{\substack{\triangle_1, \triangle_2 \in T(G) \\ \triangle_1 \neq \triangle_2}} X_{\triangle_1} X_{\triangle_2}\Big)\\
=& E\Big(\sum_{\triangle \in T(G)} X_\triangle^2\Big) + E\Big(2 \sum_{\{\triangle_1, \triangle_2\} \in D(G)} X_{\triangle_1} X_{\triangle_2}\Big)\\
&+ E\Big(2\sum_{\{\triangle_1, \triangle_2\} \in B(G)} X_{\triangle_1} X_{\triangle_2}\Big) + E\Big(2\sum_{\{\triangle_1, \triangle_2\} \in \ddot{T}(G)} X_{\triangle_1} X_{\triangle_2}\Big)\\
=& p^3 T + 2 p^4 D + 2 p^5 B+ 2 p^6 \ddot{T}\\
\leq& p^3 T + 2 (B + D) p^4 T+ 2 p^5 \tmax T + p^6 T^2\\
\leq& p^3 T + 3 p^4 \tmax T + p^6 T^2
\end{align}
where we have used that $D + B \leq \frac{3}{2} \tau T$ (by \Cref{lem:diamond_butterfly_bound}) and $2\ddot{T} \leq T^2$. At the same time, $E(T') = p^3 T$. Therefore,
\begin{align}
Var(T') = E({T'}^2) - E(T')^2 \leq p^3 T + 3 p^4 \tmax T
\end{align}
This means that $E(\hat{T}) = T$ and $Var(\hat{T}) = Var(\frac{1}{p^3}T') = \frac{1}{p^3} T + \frac{3}{p^2}\tmax T$. Since $p = \max(5 \frac{1}{\sqrt[3]{\epsilon^2 T}}, 10 \frac{1}{\epsilon \sqrt{T/\tmax}})$, it holds $Var(\hat{T}) \leq (1/5^3 + 3/10^2) \epsilon^2 T^2 < \frac{1}{20} \epsilon^2 T^2$. It, therefore holds by the Chebyshev inequality that
\[
P(|\hat{T} - T| \geq \epsilon T) \leq 1/20
\]
\end{proof}

\subsection{Triangle-heavy subgraph}
We now show an algorithm that approximately counts triangles that contain at least one vertex from set $V_H$ where $V_H$ is some given set that does not contain any triangle-light vertices.

\begin{algorithm}
\For{$v \in V_H$}{
    $\hat{T}_v \leftarrow 0$\\
    \RepTimes{$360 \frac{n^2 \log n}{\epsilon^2 \tmax}$}{
        Sample $u,w \in V$\\
        Let $\ell = |\{u,v,w\} \cap V_H\}|$\\
        If $uvw$ forms a triangle, increment $\hat{T}_v$ by $\frac{\epsilon^2 \tmax}{360 \ell    \log n}$
    }
}
\Return{$\sum_{v \in V_H} \hat{T}_v$}

\caption{Count $(1\pm\epsilon)$-approximately triangles in $G$ containing a vertex from $V_H$} \label{alg:heavy_sampling}
\end{algorithm}
\begin{lemma} \label{lem:heavy_sampling}
Given a set $V_H$, \Cref{alg:heavy_sampling} returns an unbiased estimate of the number of triangles containing at least one vertex from $V_H$.

Assume $V_H$ contains all triangle-heavy vertices of $G$ and no triangle-light vertices. Then \Cref{alg:heavy_sampling} returns $(1\pm \epsilon)$-approximation of the number of triangles that contain at least one triangle-heavy vertex with probability at least $1-O(\frac{1}{n})$. It runs in time $O\big(\frac{T n^2 \log n}{\epsilon^2 \tmax^2}\big)$. 
\end{lemma}
\begin{proof}
%
We introduce charges on vertices. For any triangle $\triangle$ that contains at least one vertex that is in $V_H$, we divide and charge single unit to the vertices in $\triangle \cap V_H$, dividing the charge fairly (if there are, e.g., two vertices from $V_H$ in the triangle, they are both charged $1/2$). Let $\chi_v$ be the charge on vertex $v$. The total amount charged (that is, $\sum_{v \in V_H} \chi_v$) is equal to the number of triangles that contain at least one vertex from $V_H$. Note that this is the quantity we want to estimate.

The algorithm considers $v \in V_H$ and samples one pair of vertices $u,w$. Consider the amount charged by $uvw$ to $v$; call it $\chi'_v$. It holds that $E(\chi'_v)= \chi_v/n^2$. When $uvw$ form a triangle, the algorithm increments $\hat{T}_v$ by $\frac{\epsilon^2 \tmax}{360 \ell \log n} = \frac{\chi'_v \epsilon^2 \tmax}{360 \log n}$. Therefore, the increment is, in expectation, equal to $\frac{\chi_v \epsilon^2 \tmax}{360 n^2 \log n}$. Since there are $\frac{n^2 \log n}{360 \epsilon^2 \tmax}$ repetitions, $v$ is in expectation charged total of
\[
\frac{360 n^2 \log n}{\epsilon^2 \tmax} \cdot \frac{\chi_v \epsilon^2 \tmax}{360 n^2 \log n} = \chi_v
\]
as we desire.

We now prove concentration around mean. For each vertex, we have $\frac{360 n^2 \log n}{\epsilon^2 \tmax}$ independent random variables corresponding to the increments in $\hat{T}$. These random variables have values between $0$ and $\frac{\epsilon^2 \tmax}{360 \log n}$ and their sum has expectation $\chi_v \geq \frac{T_v}{3} \geq \frac{\tau}{60}$ where the second inequality holds from the assumption that $V_H$ contains no triangle-light vertices. By the Chernoff bound, 
\[
P(|\hat{T}_v - T_v| \geq \epsilon T_v) \leq 2\exp\Big(-\frac{\epsilon^2 \tau/60}{3 \epsilon^2 \tau / (360 \log n)}\Big) \leq \frac{2}{n^2}
\]
By the union bound, it holds for all $v$ that $\hat{T}_v$ is an $(1\pm\epsilon)$-approximate of $T_v$, with probability at least $1-\frac{2}{n}$. On this event, the algorithm correctly gives a $(1\pm\epsilon)$-approximation of the number of triangles containing at least one triangle-heavy vertex.
\end{proof}

\subsection{Finding the triangle-heavy subgraph} \label{sec:finding_heavy}
In this section, we show how to find a set of vertices that contains each triangle-heavy vertex with probability at least $2/3$ and each triangle-light vertex with probability at most $1/3$. This guarantee may be strengthened by probability amplification to make sure that, with high probability, all triangle-heavy vertices are reported, and none of the triangle-light vertices are. This only adds $O(\log n)$ factor to the time complexity of this subroutine.

We solve separately the case $\tmax \leq n$ and $\tmax \geq n$. On the range $\tmax \leq n$, we get an algorithm running in time $\tilde{O}(\frac{n^\omega}{\tmax^{\omega-2}})$. On the range $\tmax \geq n$, we get time complexity $O(\frac{n^3}{\tmax}) \subseteq \tilde{O}(\frac{n^\omega}{\tmax^{\omega-2}})$ (the inclusion holds on this range of $\tmax$). This means that on this range, our bound is not tight. However, this is the case only on the range of $T$ for which a near-optimal algorithm was already known, and we only show this for completeness; the reader may wish to skip the case of $\tau \geq n$. Putting the guarantees for the two ranges together, it follows that
\begin{corollary} \label{cor:distinguish_heavy}
There is an algorithm, that with probability at least $1-O(\frac{1}{n})$, reports all triangle-heavy vertices and no triangle-light vertices while having time complexity $\tilde{O}(\frac{n^\omega}{\tmax^{\omega-2}})$.
\end{corollary}

\subsubsection{The case \texorpdfstring{$\tmax \leq n$}{tau <= n}}
\begin{algorithm}
$\mathcal{M} \leftarrow \emptyset$\\
\For{$i \in [k=6 \tmax^{2}]$}{
    Keep each vertex with probability $p = \tmax^{-1}$\\
    Let $V_i$ be the set of remaining vertices\\
    Find all triangles in $G[V_i]$ in time $\tilde{O}(|V_i|^\omega)$\\
    $\mathcal{M} \leftarrow \mathcal{M} \cup$ ``vertices that are contained in some triangle in $G[V_i]$''
}
\Return{$\mathcal{M}$}

\caption{Distinguish triangle-heavy vertices from triangle-light vertices} \label{alg:distinguish_matrix_multiplication}
\end{algorithm}

\begin{lemma}
\Cref{alg:distinguish_matrix_multiplication} lists any triangle-heavy vertex with probability at least $2/3$ and any triangle-light vertex with probability at most $1/3$. It has time complexity $\tilde{O}(\frac{n^\omega}{\tmax^{\omega-2}})$.
\end{lemma}
\begin{proof}
Consider a triangle-light vertex $v$. That is, there are at most $\tmax/20$ triangles in $G$ containing $v$. The probability that a triangle is retained in one iteration is, by the union bound, at most $\tmax p^3/20$. Taking a union bound over the $k$ iterations, the probability that a triangle containing $v$ is found (and thus $v$ reported) is at most $k \tmax p^3 /20 < 1/3$.

Consider a triangle-heavy vertex $v$. Let $X_i$ for $i\in[k]$ be the number of triangles containing $v$ retained in the $i$-th iteration. For $\triangle \in T$ such that $v \in \triangle$, let $X_{\triangle, i}$ be an indicator that triangle $\triangle$ is retained in $i$-th iteration. It now holds $X_i = \sum_{\triangle\in T, v\in \triangle} X_{\triangle, i}$. Let $X = \sum_{i=1}^{k} X_i$.

\medskip \noindent
By the linearity of expectation, it holds that
\[
E(X) = E\Big(\sum_{i=1}^k X_i\Big) = k p^3 T_v = 6 T_v/\tmax
\]
We now bound the variance of $X$. We first bound
\[
Var(X_i) \leq E(X_i^2) = E\bigg(\sum_{\triangle_1, \triangle_2 \in T_v} X_{\triangle_{1},i} X_{\triangle_{2},i}\bigg) \leq p^3 T_v + p^4 T_v^2
\]
Since the iterations are independent, it holds that
\[
Var(X) = Var\Big(\sum_{i=1}^{k} X_i\Big) = k Var(X_i) \leq 6 \tmax^2 (p^3 T_v + p^4 T_v^2) = \frac{6}{\tmax} T_v + \frac{6}{\tmax^2} T_v^2 \leq \frac{12 T_v^2}{\tmax^2}
\]
where the last inequality holds because $T_v \geq \tmax$. It, therefore, holds by the Chebyshev inequality that
\[
P(X = 0) \leq \frac{Var(X)}{E(X)^2} \leq \frac{12 T_v^2/\tmax^2}{(6 T_v/\tmax)^2} = 1/3
\]

We now argue the time complexity. It holds that $|V_i| \sim Bin(n,p)$. By \Cref{lem:binom_moment}, each iteration has expected time complexity $\tilde{O}(n^\omega p^\omega)$. There are $p^{-2}$ iterations, meaning that the total time complexity is
\[
\tilde{O}(p^{-2} p^\omega n^\omega) = \tilde{O}\Big(\frac{n^\omega}{\tau^{\omega - 2}}\Big)
\]
\end{proof}

\subsubsection{The case \texorpdfstring{$\tmax > n$}{tau > n}}
We repeat this case only applies on the range where an optimal algorithm is already known (see the beginning of \Cref{sec:finding_heavy}) and we only show this for completeness; the reader may wish to skip this part.
\begin{algorithm}
$\mathcal{M} \leftarrow \emptyset$\\
\For{$v \in V$}{
    \RepTimes{$k = 2n^2/\tmax$}{
        Sample $u,w \in V$\\
        \If{$uvw \in T$}{
            $\mathcal{M} \leftarrow \mathcal{M} \cup \{v\}$
        }
    }
}
\Return{$\mathcal{M}$}

\caption{Distinguish triangle-heavy vertices from triangle-light vertices} \label{alg:distinguish_sampling}
\end{algorithm}

\begin{lemma}
\Cref{alg:distinguish_sampling} lists any triangle-heavy vertex with probability at least $2/3$ and any triangle-light vertex with probability at most $1/3$. It runs in time $\tilde{O}(\frac{n^3}{\tmax})$.
\end{lemma}

\begin{proof}
Consider the probability a triangle-light vertex $v$ is reported. Similarly to the case $\tmax \leq n$, it holds by the union bound over all incident triangles and over the $k$ iterations that this probability is at most $\tmax/20 \times k \times 1/n^2 \leq 1/3$.

Consider now the probability a triangle-heavy vertex $v$ is reported. In each iteration, the probability that we find an incident triangle is $T_v/n^2 \geq \tmax/n^2$. The probability that we report $v$ is now at least
\[
1-(1-\tmax/n^2)^{2 n^2/\tmax} \geq 1-1/e^2 \geq 2/3
\]
\end{proof}

\subsection{Recursive algorithm}
We now take the subroutines presented above and put them together into recursive one algorithm. Our proof works by induction on the depth of recursion. For this reason, the statement assumes the input graph can be random. This algorithm requires advice $\tilde{T}$ such that $\tilde{T} \leq E(T)$ (note that $T$ is a random variable); we will remove the need for advice later. As we already mentioned, for the advice removal, we will need to prove a bound on the expectation of the estimate so that we are able to remove this advice $\tilde{T}$.

We show an algorithm that gives additive approximation. The reason is that this is better suited for performing recursive calls. Specifically, the time complexity of getting relative approximation worse when $T$ is small. If it happened that most triangles contained a triangle-heavy vertex, we would recurse on a subgraph with few triangles, making it more costly to get the relative approximation. We then show how this algorithm can be turned into an algorithm that gives relative approximation.

The algorithm takes a parameter $q > 0$. One can think of $q$ as being a positive number close to $0$. Decreasing this parameter improves the asymptotic time complexity of the algorithm but increases the hidden constant factor. However, this concerns a term in the time complexity that is a lower order term for constant $\epsilon$. The exact choice of $q$ is therefore not a big concern for most choices of parameters.

\begin{algorithm}
\If{$G$ has $O(1)$ vertices or $A^2 \leq \frac{10^7}{q} \tilde{T}$}{ \label{line:high_precision_case}
    \Return{number of triangles in $G$, computed using matrix multiplication} \label{line:count_exactly}
}
\If{$\tilde{T} < 1/5$}{\label{line:no_triangles}
    \Return 0 
}
$\tilde{T}' = 20\tilde{T}$\\
$\tmax \leftarrow  \frac{A^2 q^2}{40000 T}$\\
$V_H \leftarrow$ triangle-heavy vertices w.r.t. $\tmax$ using \Cref{cor:distinguish_heavy}
\label{line:distinguish}\\
$\hat{T}_H \leftarrow$ estimate number of triangles with non-empty intersection with $V_H$ using \Cref{alg:heavy_sampling} with $\epsilon = qA/2\tilde{T}'$\\ \label{line:heavy_estimate}
$V' \leftarrow$ sample each vertex from $V \setminus V_H$ independently with probability $p = \frac{20 \sqrt{\tmax \tilde{T}'}}{q A} = 1/10$ \label{line:sample_light}\\
$\hat{T}_{V'} \leftarrow$ estimate number of triangles in $G[V']$ by a recursive call with $A = (1-q) p^3 A$ and with $\tilde{T} = \tilde{T}p^3$; amplify success probability to $19/20$ by taking median of $7$ independent executions \label{line:amplify}\\
\Return{$\hat{T}_H + \hat{T}_{V'}/p^3$}

\caption{Estimate the number of triangles in $G$, given some $q > 0$ and advice $\tilde{T}$} \label{alg:tc_dense}
\end{algorithm}
Note that \cref{line:sample_light} can be efficiently implemented as follows: We pick $X \sim Bin(n,p)$, then sample $X$ vertices without replacement, keep only the vertices from $V \setminus V_H$.

\begin{lemma} \label{lem:tc_dense}
Suppose $G$ is a (possibly random) graph. Given parameters $A$, $\tilde{T}$ and $\frac{1}{2}>q>0$, \Cref{alg:tc_dense} returns $\hat{T}$ such that $E(\hat{T}) \leq E(T)$ \footnote{Since the graph can be random, $T$ is a random variable}. Moreover, if $\tilde{T} \geq E(T)$, \Cref{alg:tc_dense} returns $\hat{T} \in T \pm A$ with probability at least $4/5$. It runs in expected time $\tilde{O}\big(\frac{n^\omega}{(A^2/\tilde{T})^{\omega-2}} + \frac{\tilde{T}^{4 + 1/3 + \delta} T n^2}{A^6}\big)$ for any $\delta>0$ for $q>0$ sufficiently small.
\end{lemma}
\begin{proof}
The structure of the proof is as follows. We now
 prove that $E(\hat{T}) \leq E(T)$. We prove this by induction. We then go on to prove correctness, also by induction. We then prove the running time; in this part, we analyze the whole recursion tree instead of using induction.

\paragraph{Estimate's expected value.}
We now prove $E(\hat{T}) \leq E(T)$. We prove this by induction on the depth of recursion. If the condition on \cref{line:high_precision_case} is satisfied, then $\hat{T} = T$ and the inequality thus holds.  If the condition on \cref{line:no_triangles} is satisfied, then $\hat{T} = 0$ and the inequality also holds. Consider now the case when neither of these conditions are satisfied. By the inductive hypothesis, $E(\hat{T}_{V'}|V') \leq E(T_{V'}|V')$. Moreover, by \Cref{lem:sampling_concentration}, $E(T_{V'}/p^3|V_H) = E(T(G[V \setminus V_H])|V_H)$ and therefore $E(\hat{T}_{V'}/p^3|V_H) \leq E(E(T(G[V \setminus V_H])|V')|V_H) = E(T(G[V \setminus V_H])|V_H)$. We now bound $E(\hat{T}_H|G)$. It follows from \Cref{lem:heavy_sampling} that $E(\hat{T}_H|G)$ is an unbiased estimate of the number of triangles in $G$ containing at least one vertex from $V_H$. This implies that $E(\hat{T}_H|G) = E(T - T(G[V\setminus V_H])|G)$. We now put this all together:
\begin{align}
E(\hat{T}) &= E\big(E(\hat{T}_H|G) + E(\hat{T}_{V'}/p^3|V_H)\big) \\ &\leq E\big(E(\hat{T}_H|G)\big) + E\big(E(T(G[V\setminus V_H])|V_H)\big) \\ &= E\big(E(T-T(G[V\setminus V_H])|G)\big) + E\big(E(T(G[V\setminus V_H])|V_H)\big) \\ &= E\big(T-T(G[V\setminus V_H])\big) + E\big(T(G[V\setminus V_H])\big) \\
&= E(T)
\end{align}
Note that this bound on the expectation is not using in any way that $V_H$ contains all triangle-heavy vertices and no triangle-light vertices.

\medskip \noindent
\paragraph{Correctness.} We first focus on the base case. Consider the case $A^2 \leq \tilde{T}$. In this case, the condition on \cref{line:high_precision_case} is satisfied. The algorithm is then clearly correct.

We now focus on the case when the condition on \cref{line:no_triangles} is satisfied. We consider one call of the algorithm in the tree of recursion where this is the case.
We denote by $G,T,n,A,\tilde{T},\epsilon$ the respective values in this specific call.
We use subscript ${}_0$ to denote the respective values in the original call. For example, $T_0$ would be the number of triangles in the whole graph on which the algorithm was executed.

It holds that $E(T\wask) \leq 1/10^{3k} T_0 \leq 1/10^{3k} \tilde{T}_0 = \tilde{T}\wask$. When $\tilde{T}\wask < 1/5$ (this is the condition on \cref{line:no_triangles}), it holds by the Markov's inequality that $P(5 \tilde{T}\wask \geq T\wask) \geq 4/5$. When $\tilde{T}\wask < 1/5$, this means that $P(T\wask = 0) \geq 4/5$, meaning that the algorithm is correct in this case and runs in time $O(1)$.

We now prove the inductive step.
We now consider one recursive call, regardless of the level of recursion. Again, we use $G,T,n,A,\tilde{T},\epsilon$ to denote the respective values.
We condition on $\tilde{T}'\wask \geq T\wask$ in this call. This holds with probability at least $19/20$ by the argument based on Markov inequality written above.
We now bound the error probability by $3/20$. Together with the fact that $P(\tilde{T}'\wask \geq T) \geq 19/20$, this gives a bound on the probability of the call returning an invalid answer of $4/20$.
Let $\epsilon'\wask = \frac{qA\wask}{2\tilde{T}'\wask}$. It then holds $\epsilon' \leq \frac{qA\wask}{2 T\wask}$. We have $p = 1/10$. We show that $p \geq \max (5 \frac{1}{\sqrt[3]{{\epsilon'}^{2}\wask T\wask}}, 10 \frac{1}{{\epsilon'}\wask \sqrt{T\wask / \tau\wask}})$ (this is the assumption of \Cref{lem:sampling_concentration}). It holds $\frac{1}{10} \geq 10 \frac{1}{{\epsilon'}\wask \sqrt{T\wask / \tau\wask}}$ from the way we set $\tau\wask$. Because $A^2\wask \geq \frac{10^7}{q} \tilde{T}\wask$, it also holds
\[
5 \frac{1}{\sqrt[3]{\epsilon^{2}\wask T\wask}} \leq \frac{5}{\sqrt[3]{q A\wask^2 / (4 \tilde{T}'\wask)}} \leq \frac{5}{\sqrt[3]{q A\wask^2 / (4 \tilde{T}'\wask)}} \leq \frac{5}{\sqrt[3]{q \frac{10^7}{q} \tilde{T}\wask / (4 \tilde{T}'\wask)}} = \frac{1}{10} 
\]
By \Cref{lem:sampling_concentration}, it holds with probability at least $19/20$ that 
\[
T(G[V'])/p\wask^3 \in (1\pm \epsilon)T(G[V \setminus V_H]) \subseteq T(G[V \setminus V_H]) \pm q A\wask/2
\]
By the induction hypothesis, it holds that with probability at least $19/20$, $\hat{T}_{V'} = T(G[V']) \pm (1-q) p^3 A$. Note that $E(T_{k+1}) = E(E(T_{V'}|G)) \leq E(T\wask/10^3) \leq \tilde{T}\wask/10^3 = \tilde{T}_{k+1}$ and therefore the assumption on $\tilde{T}\wask$ is satisfied. Also note that to get success probability $19/20$, it suffices to take the median of $7$ (as can be easily checked, $P(Y \geq 5) < 1/20$ for $Y \sim Bin(7,1/5)$). Putting this together, with probability at least $18/20$, it holds that 
\[
\hat{T}_{V'}/p\wask^3 \in T(G[V'])/p\wask^3 \pm (1-q) A = T(G[V \setminus V_H]) \pm (1-q/2)A
\]
Moreover, with probability at least $19/20$, $\hat{T}_H = T(V_H) \pm \epsilon\wask T(V_H) \subseteq T(V_H) \pm q A\wask/2$ by \Cref{lem:heavy_sampling} where the inclusion holds because we have set $\epsilon = q A\wask/2\tilde{T}'\wask \leq q A\wask/2T\wask \leq q A\wask/2T\wask(V_H)$. By the union bound, with probability at least $17/20$, $|\hat{T}_H - T\wask(V_H)| \leq q A\wask/2$ and $|T(V')/p\wask^3 - T(V\wask\setminus V_H)| \leq (1-q/2)A$, in which case $\hat{T}_H + \hat{T}_L/p^3 = T \pm A$ --- the resulting answer is correct. Including the probability that $\tilde{T}' \leq T$ (we previously conditioned on this not being the case), we get that the answer is correct with probability at least $16/20$

\paragraph{Time complexity.} We again consider one call and use  $G,T,n,A,\tilde{T},\epsilon$ to denote the respective values in this one call. We first show that we may ignore in the analysis the time spent on \cref{line:count_exactly}. Counting triangles exactly by using fast matrix multiplication takes $\tilde{O}(n^\omega) \subseteq \tilde{O}\big(\frac{n^\omega}{\tau^{\omega-2}} + \frac{\tilde{T}^4 T n^2}{A^6}\big)$ time (the inclusion holds thanks to the assumption \smash{$A^2 \leq \tilde{T}$} which is satisfied on \cref{line:count_exactly}). This is, asymptotically, the same time as the bound we will use on time spent in the parent call which executed this call. We charge the time to the parent call. Since each call only makes a constant number of further recursive calls, the amount charged to the parent call is no larger than our bound on the parent call's time complexity. In the rest of the proof, we therefore ignore the time spent on \cref{line:count_exactly}.

The time complexity of \cref{line:distinguish} is $\tilde{O}(\frac{n^\omega}{\tmax^{\omega - 2}})$ and complexity of \cref{line:heavy_estimate} is $\tilde{O}(\frac{T n^2}{\epsilon^2 \tmax^2})$ by \Cref{cor:distinguish_heavy} and \Cref{lem:heavy_sampling}, respectively (note that $\epsilon$ is defined on \cref{line:heavy_estimate}). Since $p = 1/10$, the number of vertices in the $k$-th recursive call is stochastically dominated by $Bin(n,1/10^k)$ \footnote{
This is the case as in each level of recursion, we keep each triangle-light vertex with probability $p=1/10$ while removing all other vertices with high probability.} while the value of $A$ is in the $k$-th recursive call multiplied by factor of $\big((1-q)p^3\big)^{k} = \big((1-q)/10^3\big)^{k}$ and $\tilde{T}'\wask = 1/10^{3k} \tilde{T}'$.
 In each recursive call, we make $7$ recursive calls for the probability amplification. The number of calls at recursion depth $k$ is then $7^k$.
The expected time spent in the depth $k$ recursive calls on \cref{line:distinguish} is
\[
E\Big(7^k \frac{n\wask^\omega \log n\wask}{{\tmax}\wask^{\omega - 2}}\Big) \leq \tilde{O}\Big(7^k \frac{n_0^\omega/10^{\omega k}}{(\epsilon_0^2 (1-q)^k * \tilde{T}_0'/10^{3 k})^{\omega-2}}\Big) = \tilde{O}\bigg(\Big(\frac{7 \cdot 10^{3 (\omega -2)}}{(1-q)^{\omega-2} 10^{\omega}}\Big)^k \cdot \frac{n_0^\omega }{\epsilon_0^2 T_0^{\omega - 2}}\bigg)
\]
This decreases exponentially since $\omega \leq 2.5$ and, therefore, the time on \cref{line:distinguish} is dominated by the first call. (By changing constants in the algorithm, this argument can be made to work even for asymptotically slower sub-cubic matrix multiplication algorithms.) However, the time spent on \cref{line:heavy_estimate} is 
\begin{align}
E\Big(7^k \frac{T\wask n\wask^2 \wask}{\epsilon\wask^2 {\tmax}\wask^2}\Big) \leq& 7^k \frac{{\tilde{T'}\wask}^4 E(T\wask) n\wask^2}{A\wask^6}\\ \leq& 7^k \frac{{{\tilde{T'}_0}^4 }/10^{4\cdot3k} T_0/10^{3k} n^2/10^{2 k}}{A_0^6 (1-q)^{6 k}/10^{6\cdot 3k}}\\ =& \Big(\frac{10}{1-q} \Big)^k \cdot \frac{{\tilde{T'}_0}^4 T_0 n_0^2 }{A_0^6}
\end{align}
which increases at an exponential rate. We now upper-bound the number of recursive calls (in other words, we upper bound $k$). In each subsequent recursive call, $\tilde{T}$ decreases by a factor of $1/p^3 = 10^3$. This means that after $\log_{1000} T + O(1)$ recursive calls, it will hold that $\tilde{T} \leq 1/20$, in which case the algorithm finishes on \cref{line:no_triangles} in time $O(1)$ with no additional recursive calls. Therefore $k \leq \log_{1000} \tilde{T} + O(1)$. Since the amount of work at each level of recursion increases exponentially, the work is dominated by the last level of recursion. Thanks to our bound on $k$, the time complexity is
\[
\Big(\frac{10}{1-q} \Big)^{\log_{1000} \tilde{T}_0 + O(1)} \cdot \frac{{\tilde{T}_0}^4 T_0 n^2}{A_0^6} = O\Big((1-q)^{-\log_{1000} \tilde{T}_0}\frac{{\tilde{T}_0}^{4+1/3} T_0 n^2}{A_0^6}\Big) = O\Big(\frac{{\tilde{T}_0}^{4+1/3 + \log_{1000}(1/(1-q))} T_0 n^2}{A_0^6}\Big)
\]
The $\delta = \log_{1000}(1/(1-q))$ can be made arbitrarily small by picking $q$ small enough.
\end{proof}

\noindent
We now easily get the following claim, which guarantees relative approximation, unlike \Cref{lem:tc_dense}.
\begin{proposition} \label{prop:tc_dense_with_advice}
There is an algorithm that, given parameters $\epsilon,\delta > 0$ and $\tilde{T}$ returns an estimate $\hat{T}$ of $T$, such that $E(\hat{T}) \leq E(T)$ of $T$ and if, moreover, $T \leq \tilde{T} \leq 2 T$, it holds $\hat{T} \in (1 \pm \epsilon) T$ with probability at least $4/5$. It runs in expected time $\tilde{O}(\frac{n^\omega}{(\epsilon^2 T)^{\omega-2}} + \frac{n^2}{\epsilon^6 T^{2/3-\delta}})$.
\end{proposition}
\begin{proof}
By substituting $A = \epsilon \tilde{T} /2$, it is sufficient to calculate approximation with additive error $A$. This can be done by \Cref{lem:tc_dense}, giving us the desired bounds.
\end{proof}

\subsection{Removing advice} \label{sec:removing_advice}
We remove the dependency on $\tilde{T}$ by performing a geometric search. This method has been used many times before, for example in \cite{Eden2017,Goldreich2006,Aliakbarpour2018,Eden2016,Assadi2018,Eden2018}. In our case, it is slightly more complicated in that our algorithm requires both a lower and an upper bound on $T$. For this reason, we describe the method in completeness. We also prove a variant that gives a guarantee of absolute approximation; we will need this for triangle counting in sparse graphs.

\begin{theorem} \label{thm:tc_dense}
For any fixed $\delta > 0$, there is an algorithm that, given $\epsilon>0$ (respectively $A$) returns $\hat{T} \in (1 \pm \epsilon) T$ (respectively $\hat{T} \in T \pm A$) with probability at least $2/3$. It runs in expected time $\tilde{O}(\frac{n^\omega}{(\epsilon^2 T)^{\omega-2}} + \frac{n^2}{\epsilon^6 T^{2/3-\delta}})$ (respectively $\tilde{O}(\frac{n^\omega}{(A^2/\tilde{T})^{\omega-2}} + \frac{T^{5 + 1/3 + \delta} n^2}{A^6})$).
\end{theorem}
\begin{proof}
We present the argument for relative approximation. The exact same argument gives the absolute approximation by using \Cref{lem:tc_dense} in place of \Cref{prop:tc_dense_with_advice}.

We start with $\tilde{T} = \binom{n}{3}$ and, in each step, divide $\tilde{T}$ by a factor of two. When it holds that $\hat{T} \geq \tilde{T}$ where $\hat{T}$ is the estimate returned by the algorithm, we return $\hat{T}$ as the final estimate of $T$.

We now argue correctness. The estimate $\hat{T}$ given by the algorithm is always non-negative and it holds $E(\hat{T}) \leq T$. By Markov's inequality, $P(\hat{T} \geq 2T) \leq 1/2$, regardless of the choice of $\tilde{T}$. We amplify this probability to $\Theta(1/\log n)$ by taking the median of $\Theta(\log \log n)$ independent executions. When $\hat{T} \geq \tilde{T}$ (this is when we return $\hat{T}$ as the final estimate), it holds with probability at least $1-O(1/\log n)$ that $\tilde{T} \leq 2 T$ as otherwise, it would be the case that $\hat{T} \geq \tilde{T} \geq 2T$ which holds with probability $O(1/\log n)$. Conditioned on $\tilde{T} \leq 2 T$, the algorithm gives a $(1\pm \epsilon)$-approximate estimate (respectively additive $\pm A$ estimate) by \Cref{prop:tc_dense_with_advice} (respectively \Cref{lem:tc_dense}). By the union bound, the failure probability is then arbitrarily small if we make the constants in $\Theta$ sufficiently small.



We now argue the time complexity. The probability amplification only increases the time complexity by a $O(\log\log n)$ factor. When $\tilde{T} \leq T$, with probability $1-O(1/\log n)$ we return the estimate and quit. We now consider the calls of the algorithm with $\tilde{T} \leq T$. The time complexity increases exponentially (as $T$ is decreasing exponentially) while the probability of performing a call is decreasing at a rate faster than exponential. Therefore, the time complexity is dominated by the first call when $\tilde{T} \leq T$. This time complexity is as claimed by \Cref{prop:tc_dense_with_advice} (respectively \Cref{lem:tc_dense}).
\end{proof}
\subsection{Optimality}
We now show that the lower bound of \citet{Eden2017} implies that our algorithm is in certain sense optimal. Specifically, we show that the exponent in the dependency on $T$ cannot be improved without worsening the dependency on $n$, as long as the exponent of $T$ is constant. That still leaves open the possibility of improving the dependency on $n$ and getting an algorithm with a non-constant exponent of $T$.
\begin{proposition}
Suppose there is an algorithm which runs in time $\tilde{O} \big(\frac{n^\omega}{T^\delta} \big)$ for some constant $\delta$, and returns a constant-factor approximation of $T$ with probability at least $2/3$. Then $\delta \leq \omega-2$.
\end{proposition}
\begin{proof}
Suppose $\delta > \omega-2$. Then, for $T = \Theta(n)$, the algorithm runs in time $o(n)$. This is in contradiction to the lower bound of \citet{Eden2017} which states that any such algorithms has to run in time $\Omega(n)$.
\end{proof}

\section{Triangle counting in sparse graphs} \label{sec:sparse}

We now show how \Cref{alg:tc_dense} can be used to get an efficient algorithm for counting triangles in a sparse graph. The algorithm finds all vertices with degree at least some parameter $\theta$, then it uses \Cref{alg:tc_dense} to count triangles in the subgraph induced by these vertices and uses sampling in the rest of the graph. For this, we first need to find all vertices with degree at least $\theta$. 
We set $\theta = \frac{m^{({\omega -1})/({\omega +1})}}{T^{({\omega -3})({\omega +1})} \epsilon^{({2\omega -6})/({\omega +1})}}$.

\subsection{Finding high-degree vertices}
We now show how to find all vertices with degree at least $\theta$. However, note that the contribution of this paper is mainly in the superlinear regime where all such vertices can be found by iterating over all vertices. \emph{The reader may, therefore, want to skip this section.}

\begin{algorithm}
$S \leftarrow$ sample $\frac{2 m \log n}{\theta}$ edges\\
$V_H \leftarrow$ vertices with degree $\geq \theta$ that are an endpoint of some vertex in $S$\\
\Return{$V_H$}

\caption{Return all vertices with degree at least $\theta$} \label{alg:find_high_degrees}
\end{algorithm}
\begin{lemma}
\Cref{alg:find_high_degrees} returns all vertices with degree at least $\theta$ with probability at least $1-\frac{1}{n}$. It runs in time $O(\frac{m \log n}{\theta})$
\end{lemma}
\begin{proof}
The time complexity is clearly as claimed. We now prove correctness. We need to prove that with probability at least $1-\frac{1}{n}$, all vertices with degree at least $\theta$ have at least one of their incident edges in $S$. Consider one such vertex $v$ and pick an edge $e$ at random. It holds that $P(v \in e) = \frac{d(v)}{m} \geq \theta/m$. We, therefore, have $\frac{2 m \log n}{\theta}$ trials, each having a success probability at least $\theta/m$. The probability that none of them succeeds is then at most
\[
(1-\theta/m)^{2 m \log(n)/\theta} \leq e^{-2 \log n} = \frac{1}{n^2}
\]
Taking union bound over the $n$ vertices, we get that each vertex with $\geq \theta$ is an endpoint of some edge in $S$ with probability at least $1-\frac{1}{n}$.
\end{proof}

\subsection{Counting triangles}
We now show the algorithm for counting triangles in sparse graphs. In the algorithm, we use a total order $\prec$ defined as $u \prec v$ if $d(u) \leq d(v)$ and ties broken arbitrarily in a consistent way (so that the resulting relation is a total order).

\begin{algorithm}
$\theta \leftarrow \frac{m^{({\omega -1})/({\omega +1})}}{T^{({\omega -3})({\omega +1})} \epsilon^{({2\omega -6})/({\omega +1})}}$\\
$V_H \leftarrow$ find all vertices $v$ with $d(v) \geq \theta$ using \Cref{alg:find_high_degrees}\\
$\hat{T}_H \leftarrow$ count triangles in $G[V_H]$ using \Cref{alg:tc_dense} with error $\pm \epsilon \tilde{T}/2$; amplify success probability to $\frac{5}{6}$\\

\medskip
$M \leftarrow 0$\\
\RepTimes{$k = 12 \frac{\theta m}{\epsilon^2 \tilde{T}}$}{ \label{line:loop}
    $e \leftarrow$ pick an edge uniformly at random\\
    $uv \leftarrow e$, such that $u \succ v$\\
    $w \leftarrow$ pick a random neighbor of $v$\\
    \If{$v \not\in V_H$ and $w \succ v$ and $uvw \in T$}{
        $M \leftarrow M + d(v)$\\
    }
}

$\hat{T}_L = \frac{m}{2 k} M$

\Return{$\hat{T}_H + \hat{T}_L$}

\caption{Estimate the number of triangles in $G$, given some $q > 0$ and advice $\tilde{T}$} \label{alg:tc_sparse}
\end{algorithm}

We now prove the main theorem for triangle counting in sparse graphs. The last part in the following theorem is there to enable us to remove advice in the same way we did in \Cref{sec:removing_advice}.
\begin{lemma} \label{lem:tc_sparse}
Let us have a graph $G$. Given parameters $1 > \epsilon >0, q > 0$ and advice $\tilde{T}$, \Cref{alg:tc_sparse} returns $\hat{T}$ such that $P(\hat{T} \geq 2 T) \leq 2/3$. Moreover, if $\tilde{T} \leq T$, \Cref{alg:tc_sparse} returns $\hat{T} \in (1\pm \epsilon)T$ with probability at least $4/5$. It runs in expected time
\[
\tilde{O}\Big(\frac{m^{2\omega/(\omega+1)}}{\epsilon^{4(\omega-1)/(\omega+1)} \tilde{T}^{2(\omega-1)/(\omega+1)}} + \frac{m^{4/(\omega+1)}}{\epsilon^{2(\omega+1)/(\omega+1)} T^{(7-\omega)/(\omega+1)-\delta}}\Big)
\]
for any $\delta>0$ for $q>0$ sufficiently small.
\end{lemma}


\noindent
With the current best bound on $\omega$, we get running time of $\tilde{O}(\frac{m^{1.408}}{\epsilon^{1.628} T^{0.814}} +\frac{m^{1.186}}{\epsilon^{6.743} T^{1.371}})$. Note that for constant $\epsilon$, the second term is significantly smaller than the first one.
\begin{proof}
Let $T_L$ be the set of triangles that are not contained in $G[V_H]$ --- or equivalently, that have their $\prec$-minimal vertex outside of $V_H$ --- and, by an abuse of notation, also the number of such triangles. Let $T_H$ be the number of triangles contained in $G[V_H]$.

Let $X_i$ be the increment in $M$ in the $i$-th iteration of the loop on line \cref{line:loop}. We now calculate the expectation of $M$ after all $k$ iterations. Let us fix a triangle $\triangle \in T$ and define $X_i' = X_i$ if $uvw = \triangle$ in the $i$-th iteration and $X_i' = 0$ otherwise (note that $u,v,w$ are defined in the algorithm). If $a$ is the vertex of $\triangle$ minimal w.r.t. $\prec$, it now holds that $E(X_i') = d(a) P(X_i' = d(a))$. It holds that $X_i' = d(a)$ if and only if $v = a$ and $uvw = \triangle$. For this to happen, the algorithm has to sample in the $i$-th iteration the edge $uv$ or $wv$. This happens with probability $2/m$. Then, it has to be the case that the random neighbor sampled on line $7$ is the single vertex of $\triangle$ not in $e$. This happens with probability $1/d(a)$. This gives us that $E(X_i') = \frac{d(a)}{m d(a)} = 2/m$.  By summing over all triangles in $T_L$ and by linearity of expectation, it holds that $E(X_i) = 2T_L/m$. Therefore, the estimate $\hat{T}_L = \frac{m}{2 k} M = \frac{m}{2 k} \sum_{i=1}^k X_i$ is an unbiased estimate of $T_L$.

We now bound the variance. We use the inequality\footnote{$\sup(X)$ is the smallest $x$ such that $P(X > x) = 0$.} $Var(X) \leq \sup(X)E(X)$ to get $Var(X_i) \leq \theta E(X_i) = 2 \theta T_L/m$. Therefore $Var(m X_i/2) = \frac{1}{2} \theta m T_L$. Taking average of $12 \frac{\theta m}{\epsilon^2 \tilde{T}}$ independent trials, we get variance $\frac{1}{24} \epsilon^2 \tilde{T} T_L \leq \frac{1}{24} \epsilon^2 T^2$. By the Chebyshev's inequality, it therefore holds that
\[
P(|\hat{T}_L - T_L| \geq \epsilon T/2) \leq \frac{Var(\hat{T}_L)}{(\epsilon T/2)^2} \leq 1/6
\]

At the same time, by \Cref{lem:tc_dense}, it holds that $\hat{T}_H \in T_H \pm \epsilon \tilde{T}/2 \subseteq T_H \pm \epsilon T/2$ with probability at least $5/6$ (note that we have amplified the success probability). By the union bound, with probability at least $2/3$, it holds that both $|\hat{T}_L - T_L| \leq \epsilon T / 2$ and $|\hat{T}_H - T_H| \leq \epsilon T/2$. Therefore, the algorithm returns a $(1\pm \epsilon)$-approximate answer with probability at least $2/3$.

We now argue the time complexity. There are no more than than $m/\theta$ vertices with degree at least $\theta$. Triangles in $G[V_H]$ are therefore counted, by \Cref{thm:tc_dense}, in time 
\begin{align}
\tilde{O}\bigg(\frac{(m/\theta)^\omega}{(A^2 /\tilde{T})^{\omega-2}} + \frac{T_H^{5 + 1/3 + \delta} (m/\omega)^2}{A^6}\bigg) \leq& \tilde{O}\bigg(\frac{(m/\theta)^\omega}{(\epsilon^2 \tilde{T})^{\omega-2}} + \frac{(m/\omega)^2}{\epsilon^6 T^{2/3-\delta}}\bigg)\\ =& \tilde{O}\Big(\frac{m^{2\omega/(\omega +1)}}{\epsilon^{4(\omega-1)/(\omega+1)} \tilde{T}^{2(\omega-1)/(\omega+1)}} + \frac{m^{4/(\omega+1)}}{\epsilon^{2(\omega+1)/(\omega+1)} T^{(7-\omega)/(\omega+1)-\delta}}\Big)
\end{align}
The time spent in each iteration of the loop on \cref{line:loop} is $O(1)$. The total time spent in the loop is therefore \[
O\Big(\frac{\theta m}{A^2/\tilde{T}}\Big) = O\Big(\frac{\theta m}{\epsilon^2 \tilde{T}}\Big) =  O\Big(\frac{m^{2\omega/(\omega +1)}}{\epsilon^{4(\omega-1)/(\omega+1)} \tilde{T}^{2(\omega-1)/(\omega+1)}}\Big)
\]
Putting these two time complexities together, we get the desired running time.

We now prove that $P(\hat{T} \geq 2 T) \leq 2/3$ \footnote{We cannot use the Markov's inequality in a straightforward way. The reason is that it is not clear how to bound the expectation of the estimate coming from \Cref{thm:tc_dense}. While we do have a bound on the expectation of the estimate given by, \Cref{alg:tc_dense} it is not clear a bound of this type carries over when we perform the advice removal.}. Since $\epsilon < 1$, it holds that $P(\hat{T}_H \geq 2 T_H) \leq P(\hat{T}_H \geq (1+\epsilon) T_H) \leq 1/6$ as we have amplified the success probability of \Cref{alg:tc_dense} to $\frac{5}{6}$. Moreover, as we have shown, $E(\hat{T}_L) = T_L$ and, therefore, $P(\hat{T}_L \geq 2T_L) \leq 1/2$
\end{proof}

\medskip \noindent
By the same argument as in \Cref{sec:removing_advice}, we may remove the advice and get the following claim\footnote{Instead of directly using the Markov inequality in the argument, we may use the last part of \Cref{lem:tc_sparse} to get the same guarantee.}. 

\begin{theorem} \label{thm:tc_sparse}
There is an algorithm that returns a $(1\pm \epsilon)$-approximation of the number of triangles with probability oat least $2/3$. It runs in time
\[
\tilde{O}\Big(\frac{m^{2\omega/(\omega+1)}}{\epsilon^{4(\omega-1)/(\omega+1)} T^{2(\omega-1)/(\omega+1)}} + \frac{m^{4/(\omega+1)}}{\epsilon^{2(\omega+1)/(\omega+1)} T^{(7-\omega)/(\omega+1)-\delta}}\Big)
\]
where $\delta>0$ is a parameter that can be chosen arbitrarily.
\end{theorem}

\subsection{Optimality}
Like in the case of dense graphs, we can show that our algorithm is in certain sense optimal.

\begin{proposition}
Suppose there is an algorithm that uses both random vertex and random edge queries, which runs in time $\tilde{O}\big(\frac{m^{2\omega/(\omega+1)}}{T^\delta} \big)$ for some constant $\delta$, and returns a constant-factor approximation of $T$ with probability at least $2/3$. Then $\delta \leq 2(\omega-1)/(\omega+1)$.
\end{proposition}
\begin{proof}
Suppose $\delta > 2(\omega-1)/(\omega+1)$. Then, for $T = \Theta(\sqrt{m})$, the algorithm runs in time $o(m)$. This is in contradiction to the lower bound of \citet{Eden2017} which states that any such algorithms has to run in time $\Omega(m)$.
\end{proof}

\paragraph{Acknowledgements.}
I would like to thank several people: Václav Rozhoň, who gave early feedback on the results and who suggested to me I try the recursive approach; Rasmus Pagh for helping improve a draft of this paper; my supervisor Mikkel Thorup for helpful discussions and for his support; and Talya Eden for advice regarding related work. I am very grateful to Bakala Foundation for financial support.

\bibliographystyle{plainnat}
\bibliography{literature}

\begin{thebibliography}{19}
\providecommand{\natexlab}[1]{#1}
\providecommand{\url}[1]{\texttt{#1}}
\expandafter\ifx\csname urlstyle\endcsname\relax
  \providecommand{\doi}[1]{doi: #1}\else
  \providecommand{\doi}{doi: \begingroup \urlstyle{rm}\Url}\fi

\bibitem[Aliakbarpour et~al.(2018)Aliakbarpour, Biswas, Gouleakis, Peebles,
  Rubinfeld, and Yodpinyanee]{Aliakbarpour2018}
Maryam Aliakbarpour, Amartya~Shankha Biswas, Themis Gouleakis, John Peebles,
  Ronitt Rubinfeld, and Anak Yodpinyanee.
\newblock {Sublinear-Time Algorithms for Counting Star Subgraphs via Edge
  Sampling}.
\newblock \emph{Algorithmica}, 80\penalty0 (2):\penalty0 668--697, feb 2018.
\newblock ISSN 14320541.
\newblock \doi{10.1007/s00453-017-0287-3}.
\newblock URL \url{https://doi.org/10.1007/s00453-017-0287-3}.

\bibitem[Alman and Williams(2021)]{Alman2021}
Josh Alman and Virginia~Vassilevska Williams.
\newblock A {Refined} {Laser} {Method} and {Faster} {Matrix} {Multiplication}.
\newblock In \emph{Proceedings of the 2021 {ACM}-{SIAM} {Symposium} on
  {Discrete} {Algorithms} ({SODA})}, Proceedings, pages 522--539. Society for
  Industrial and Applied Mathematics, January 2021.
\newblock \doi{10.1137/1.9781611976465.32}.
\newblock URL \url{https://epubs.siam.org/doi/abs/10.1137/1.9781611976465.32}.

\bibitem[Alon et~al.(1997)Alon, Yuster, and Zwick]{Alon1997}
N.~Alon, R.~Yuster, and U.~Zwick.
\newblock {Finding and Counting Given Length Cycles}.
\newblock \emph{Algorithmica (New York)}, 17\penalty0 (3):\penalty0 209--223,
  1997.
\newblock ISSN 01784617.
\newblock \doi{10.1007/BF02523189}.
\newblock URL \url{https://link.springer.com/article/10.1007/BF02523189}.

\bibitem[Assadi et~al.(2019)Assadi, Kapralov, and Khanna]{Assadi2018}
Sepehr Assadi, Michael Kapralov, and Sanjeev Khanna.
\newblock {A simple sublinear-time algorithm for counting arbitrary subgraphs
  via edge sampling}.
\newblock \emph{Leibniz International Proceedings in Informatics, LIPIcs}, 124,
  nov 2019.
\newblock ISSN 18688969.
\newblock \doi{10.4230/LIPIcs.ITCS.2019.6}.

\bibitem[Barabási and Pósfai(2016)]{Barabasi2016}
Albert-László Barabási and Márton Pósfai.
\newblock \emph{Network science}.
\newblock Cambridge University Press, Cambridge, 2016.
\newblock ISBN 9781107076266 1107076269.
\newblock URL \url{http://barabasi.com/networksciencebook/}.

\bibitem[Biswas et~al.(2020)Biswas, Eden, Liu, Mitrovi'c, and
  Rubinfeld]{Biswas2020}
Amartya~Shankha Biswas, T.~Eden, Q.~Liu, Slobodan Mitrovi'c, and R.~Rubinfeld.
\newblock Parallel algorithms for small subgraph counting.
\newblock \emph{ArXiv}, abs/2002.08299, 2020.

\bibitem[Chiba and Nishizeki(1985)]{Chiba1985}
N.~Chiba and Takao Nishizeki.
\newblock Arboricity and subgraph listing algorithms.
\newblock \emph{SIAM J. Comput.}, 14:\penalty0 210--223, 1985.

\bibitem[Chung and Lu(2006)]{Chung2006}
Fan Chung and Linyuan Lu.
\newblock \emph{Complex Graphs and Networks (Cbms Regional Conference Series in
  Mathematics)}.
\newblock American Mathematical Society, USA, 2006.
\newblock ISBN 0821836579.

\bibitem[Eden and Rosenbaum(2018)]{Eden2018b}
Talya Eden and Will Rosenbaum.
\newblock {Lower Bounds for Approximating Graph Parameters via Communication
  Complexity}.
\newblock In Eric Blais, Klaus Jansen, Jos{\'e} D.~P. Rolim, and David Steurer,
  editors, \emph{Approximation, Randomization, and Combinatorial Optimization.
  Algorithms and Techniques (APPROX/RANDOM 2018)}, volume 116 of \emph{Leibniz
  International Proceedings in Informatics (LIPIcs)}, pages 11:1--11:18,
  Dagstuhl, Germany, 2018. Schloss Dagstuhl--Leibniz-Zentrum fuer Informatik.
\newblock ISBN 978-3-95977-085-9.
\newblock \doi{10.4230/LIPIcs.APPROX-RANDOM.2018.11}.
\newblock URL \url{http://drops.dagstuhl.de/opus/volltexte/2018/9415}.

\bibitem[Eden et~al.(2017{\natexlab{a}})Eden, Levi, Ron, and
  Seshadhri]{Eden2017}
Talya Eden, Amit Levi, Dana Ron, and C.~Seshadhri.
\newblock Approximately counting triangles in sublinear time.
\newblock \emph{SIAM Journal on Computing}, 46\penalty0 (5):\penalty0
  1603--1646, 2017{\natexlab{a}}.
\newblock \doi{10.1137/15M1054389}.
\newblock URL \url{https://doi.org/10.1137/15M1054389}.

\bibitem[Eden et~al.(2017{\natexlab{b}})Eden, Ron, and Seshadhri]{Eden2016}
Talya Eden, Dana Ron, and C.~Seshadhri.
\newblock {Sublinear Time Estimation of Degree Distribution Moments: The
  Degeneracy Connection}.
\newblock In Ioannis Chatzigiannakis, Piotr Indyk, Fabian Kuhn, and Anca
  Muscholl, editors, \emph{44th International Colloquium on Automata,
  Languages, and Programming (ICALP 2017)}, volume~80 of \emph{Leibniz
  International Proceedings in Informatics (LIPIcs)}, pages 7:1--7:13,
  Dagstuhl, Germany, 2017{\natexlab{b}}. Schloss Dagstuhl--Leibniz-Zentrum fuer
  Informatik.
\newblock ISBN 978-3-95977-041-5.
\newblock \doi{10.4230/LIPIcs.ICALP.2017.7}.
\newblock URL \url{http://drops.dagstuhl.de/opus/volltexte/2017/7374}.

\bibitem[Eden et~al.(2018)Eden, Ron, and Seshadhri]{Eden2018}
Talya Eden, Dana Ron, and C.~Seshadhri.
\newblock On approximating the number of k-cliques in sublinear time.
\newblock In \emph{Proceedings of the 50th Annual ACM SIGACT Symposium on
  Theory of Computing}, STOC 2018, page 722–734, New York, NY, USA, 2018.
  Association for Computing Machinery.
\newblock ISBN 9781450355599.
\newblock \doi{10.1145/3188745.3188810}.
\newblock URL \url{https://doi.org/10.1145/3188745.3188810}.

\bibitem[Goldreich and Ron(2006)]{Goldreich2006}
Oded Goldreich and Dana Ron.
\newblock Approximating average parameters of graphs.
\newblock In Josep D{\'i}az, Klaus Jansen, Jos{\'e} D.~P. Rolim, and Uri Zwick,
  editors, \emph{Approximation, Randomization, and Combinatorial Optimization.
  Algorithms and Techniques}, pages 363--374, Berlin, Heidelberg, 2006.
  Springer Berlin Heidelberg.
\newblock ISBN 978-3-540-38045-0.

\bibitem[Itai and Rodeh(1977)]{Itai1977}
Alon Itai and Michael Rodeh.
\newblock Finding a minimum circuit in a graph.
\newblock In \emph{Proceedings of the Ninth Annual ACM Symposium on Theory of
  Computing}, STOC '77, page 1–10, New York, NY, USA, 1977. Association for
  Computing Machinery.
\newblock ISBN 9781450374095.
\newblock \doi{10.1145/800105.803390}.
\newblock URL \url{https://doi.org/10.1145/800105.803390}.

\bibitem[Kallaugher and Price(2017)]{Kallaugher2016}
John Kallaugher and Eric Price.
\newblock A hybrid sampling scheme for triangle counting.
\newblock In \emph{Proceedings of the Twenty-Eighth Annual ACM-SIAM Symposium
  on Discrete Algorithms}, SODA '17, page 1778–1797, USA, 2017. Society for
  Industrial and Applied Mathematics.

\bibitem[Kolountzakis et~al.(2012)Kolountzakis, Miller, Peng, and
  Tsourakakis]{Kolountzakis2012}
Mihail~N. Kolountzakis, Gary~L. Miller, Richard Peng, and Charalampos~E.
  Tsourakakis.
\newblock Efficient triangle counting in large graphs via degree-based vertex
  partitioning.
\newblock \emph{Internet Mathematics}, 8\penalty0 (1-2):\penalty0 161--185,
  2012.
\newblock \doi{10.1080/15427951.2012.625260}.
\newblock URL \url{https://doi.org/10.1080/15427951.2012.625260}.

\bibitem[Pagh and Tsourakakis(2012)]{Pagh2012}
Rasmus Pagh and Charalampos~E. Tsourakakis.
\newblock Colorful triangle counting and a mapreduce implementation.
\newblock \emph{Inf. Process. Lett.}, 112\penalty0 (7):\penalty0 277–281,
  March 2012.
\newblock ISSN 0020-0190.
\newblock \doi{10.1016/j.ipl.2011.12.007}.
\newblock URL \url{https://doi.org/10.1016/j.ipl.2011.12.007}.

\bibitem[Tsourakakis et~al.(2009)Tsourakakis, Kang, Miller, and
  Faloutsos]{Tsourakakis2009}
Charalampos~E. Tsourakakis, U.~Kang, Gary~L. Miller, and Christos Faloutsos.
\newblock Doulion: Counting triangles in massive graphs with a coin.
\newblock In \emph{Proceedings of the 15th ACM SIGKDD International Conference
  on Knowledge Discovery and Data Mining}, KDD '09, page 837–846, New York,
  NY, USA, 2009. Association for Computing Machinery.
\newblock ISBN 9781605584959.
\newblock \doi{10.1145/1557019.1557111}.
\newblock URL \url{https://doi.org/10.1145/1557019.1557111}.

\bibitem[Wimmer and Lewis(2010)]{Wimmer2010}
Andreas Wimmer and Kevin Lewis.
\newblock Beyond and below racial homophily: Erg models of a friendship network
  documented on facebook.
\newblock \emph{American Journal of Sociology}, 116\penalty0 (2):\penalty0
  583--642, 2010.
\newblock ISSN 00029602, 15375390.
\newblock URL \url{http://www.jstor.org/stable/10.1086/653658}.

\end{thebibliography}

\end{document}